\pdfoutput=1
\documentclass[sigconf,screen]{acmart}

\setcopyright{rightsretained}
\acmPrice{}
\acmDOI{10.1145/3611643.3616290}
\acmYear{2023}
\copyrightyear{2023}
\acmSubmissionID{fse23main-p419-p}
\acmISBN{979-8-4007-0327-0/23/12}
\acmConference[ESEC/FSE '23]{Proceedings of the 31st ACM Joint European Software Engineering Conference and Symposium on the Foundations of Software Engineering}{December 3--9, 2023}{San Francisco, CA, USA}
\acmBooktitle{Proceedings of the 31st ACM Joint European Software Engineering Conference and Symposium on the Foundations of Software Engineering (ESEC/FSE '23), December 3--9, 2023, San Francisco, CA, USA}
\received{2023-02-02}
\received[accepted]{2023-07-27}

\usepackage{microtype} 

\usepackage{xspace}
\newcommand{\rev}[1]{#1\xspace}

\usepackage{amsmath}
\usepackage{coqtheorem}
\usepackage{mybnf}
\usepackage{mathpartir}
\usepackage[frozencache,cachedir=minted-cache]{minted}
\usepackage{xcolor}
\usepackage{tikz-cd}
\usetikzlibrary{matrix,shapes,arrows,positioning,chains}
\usepackage{subcaption}
\usepackage{algpseudocode}
\usepackage{xpatch}
\usepackage{epigraph}
\usepackage{glossaries}
\usepackage[inline]{enumitem}

\usepackage{cleveref}
\crefname{section}{Sect.}{Sects.}
\Crefname{section}{Section}{Sections}
\crefname{figure}{Fig.}{Figs.}
\Crefname{figure}{Figure}{Figures}
\crefname{question}{Question}{Questions}
\Crefname{question}{Question}{Questions}

\newif\ifreview
\reviewfalse

\newcommand\wise{WiSE}
\newcommand\pywise{PyWiSE}
\newcommand\cstore{V}
\newcommand\sstore{S}
\newcommand\pathcond{\varphi}
\newcommand\ltlformula{\varphi}
\newcommand\prog{p}
\newcommand\progb{q}
\newcommand\cstate{\sigma}
\newcommand\sstate{{\hat{\sigma}}}
\newcommand\stream{\alpha}
\newcommand{\always}{\text{\syn{\(\square\)}}}
\newcommand{\eventually}{\text{\syn{\(\lozenge\)}}}
\newcommand{\ltlimplies}{\text{\syn{\(\to\)}}}

\newcommand{\hasBug}{\textit{HasBug}}
\newcommand{\hasBugSym}{\textit{HasBug}_\textrm{sym}}
\newcommand{\reach}{\textit{Reach}}
\newcommand{\reachsym}{\textit{Reach}_\textrm{sym}}
\newcommand{\reachableFrom}{\textit{ReachableFrom}}
\newcommand{\validStatus}{\textit{ValidStatus}}
\newcommand{\findBugs}{\texttt{find\_bugs}}
\newcommand{\concrete}{\textit{Concrete}}
\newcommand{\symbolic}{\textit{Symbolic}}

\newcommand{\stateconcr}[1]{\simeq_{#1}}

\newacronym{se}{SE}{Symbolic Execution}
\newacronym{ai}{AI}{Abstract Interpretation}
\newacronym{ltl}{LTL}{Linear Temporal Logic}

\newenvironment{tightcenter}{%
  \setlength\topsep{0pt}
  \setlength\parskip{0pt}
  \begin{center}
}{%
  \end{center}
}

\newcommand\python[1]{\mintinline{python}{#1}}

\usepackage{fixme}
\fxusetheme{colorsig}
\fxsetup{status=draft}
\FXRegisterAuthor{ds}{ads}{DS}
\FXRegisterAuthor{ac}{aac}{AC}

\xpretocmd{\proof}{\setlength{\parindent}{0pt}}{}{}

\definecolor{mybrown}{HTML}{974006}

\renewcommand\bnfannot[1]{ #1 }
\renewcommand\bnfexpr[1]{\textit{#1}}
\newcommand\syn[1]{\texttt{\color{mybrown}{#1}}}
\newcommand\st[2]{\ensuremath{\langle #1, #2 \rangle}}
\newcommand\sst[3]{\ensuremath{\langle #1, #2, #3 \rangle}}
\newcommand\den[1]{\ensuremath{\llbracket #1 \rrbracket}}

\newcommand\code[1]{{\ttfamily #1}}

\theoremstyle{coqtheorem}
\newcoqtheorem{theorem}{Theorem}
\newcoqtheorem{lemma}{Lemma}
\newcoqtheorem{corollary}{Corollary}
\newcoqtheorem{fact}{Fact}
\newcoqtheorem{definition}{Definition}

\crefname{theorem}{Thm.}{Thms.}
\Crefname{theorem}{Theorem}{Theorems}

\setBaseUrl{{https://anonymousverifier.github.io/wise_doc/}}
\setCoqFilename{WiSE.symbolic.symex}

\author{Arthur Correnson}
\email{arthur.correnson@cispa.de}
\orcid{0000-0003-2307-2296}
\author{Dominic Steinhöfel}
\email{dominic.steinhoefel@cispa.de}
\orcid{0000-0003-4439-7129}
\affiliation{%
  \institution{CISPA Helmholtz Center for Information Security}
  \city{Saarbrücken}
  \country{Germany}
}



\begin{document}
  \title{Engineering a Formally Verified Automated Bug Finder}

  \begin{abstract}
    Symbolic execution is a program analysis technique executing programs with symbolic instead of concrete inputs.
    This principle allows for exploring many program paths at once.
    Despite its wide adoption---in particular for program \emph{testing}--little effort was dedicated to studying the semantic foundations of symbolic execution.
    Without these foundations, critical questions regarding the correctness of symbolic executors cannot be satisfyingly answered:
    Can a reported bug be reproduced, or is it a \emph{false positive} (soundness)?
    Can we be sure to find \emph{all bugs} if we let the testing tool run long enough (completeness)?
    This paper presents a systematic approach for engineering provably sound and complete symbolic execution-based bug finders by relating a programming language's operational semantics with a symbolic semantics.
    In contrast to prior work on symbolic execution semantics, we address the correctness of critical implementation details of symbolic bug finders, including the search strategy and the role of constraint solvers to prune the search space.
    We showcase our approach by implementing WiSE, a prototype of a verified bug finder for an imperative language, in the Coq proof assistant and proving it sound and complete.
    We demonstrate that the design principles of \wise{} survive outside the ecosystem of interactive proof assistants by (1) automatically extracting an OCaml implementation and (2) transforming \wise{} to \pywise{}, a functionally equivalent Python version.
  \end{abstract}


  \maketitle



  \section*{Introduction}

  \epigraph{%
    It is now two decades since it was pointed out that program testing may convincingly demonstrate the presence of bugs, but can never demonstrate their absence.
    After quoting this well-publicized remark devoutly, the software engineer returns to the order of the day and continues to refine his testing strategies, just like the alchemist of yore, who continued to refine his chrysocosmic purifications.
  }{%
    \textit{E.W.\ Dijkstra}, 1988~\cite{dijkstra-88}
  }

  A bridge that works fine for a year but collapses during a windy day is considered a failure.
  A program that works fine for a year but crashes for a certain input is considered buggy.
  This ridiculous analogy reflects the reality of software engineering:
  While civil engineers carefully identify loads and associated stresses on the structures they build and create their designs accordingly, preventing tragic crashes upfront, ``bugs'' in software are most frequently reported by \emph{users} only after deployment~\cite{rollbar-21}.
  However, software engineers are in a fortunate position:
  They can stress-test the \emph{final product} before it is shipped to users.
  This becomes an opportunity only if programmers possess adequate verification tools.
  The predominant verification discipline today is \emph{testing}, which aims to find bugs that are \emph{actionable} in the sense that running the program under test with an \emph{input} provided by the testing tool results in an error state.
  For example, \emph{evolutionary graybox fuzzers} (such as AFL++~\cite{fioraldi.maier.ea-20}) mutate inputs from an initial population and add mutants covering new parts in the program.
  While modern fuzzers work exceptionally well in many cases, the effectiveness of evolutionary fuzzers strongly depends on the quality of the initial seed inputs~\cite{liyanage.boehme.ea-23}.
  Furthermore, they struggle to cover code guarded with complex constraints.
  Consider the Python program
  \begin{tightcenter}
    \python{if x == hash((y, z)): fail()}
  \end{tightcenter}
  Finding a triple of values satisfying the \python{if} condition randomly or by mutating previous inputs is extremely difficult.
  Indeed, Bundt et al.~\cite{bundt.fasano.ea-21} demonstrated that graybox fuzzers are ``mostly useless'' for certain bugs.
  However, they state that integrating \emph{symbolic execution} with graybox fuzzers to so-called \emph{hybrid fuzzers} is \emph{highly effective,} a result also confirmed by other studies (e.g.,~\cite{boehme.paul-16}). 

  \gls{se} builds on the principle of declaring some (or all) inputs to a program to be \emph{symbolic}.
  While concrete execution maintains states mapping variables to values, \gls{se} uses \emph{symbolic states} consisting of a \emph{path condition} and a \emph{symbolic store.}
  At \emph{branching points} such as the \python{if} statement from above, \gls{se} generates multiple successor states distinguished by different constraints added to their path condition.
  As usual, \emph{assignments} update the symbolic store, which can map to symbolic values.
  \gls{se} was conceived almost simultaneously as a bug finding~\cite{king-76} and program proving technique~\cite{burstall-74} and is still actively used in both areas.
  It is more popular in testing, though, due to its \emph{precise nature:}
  Under \emph{ideal assumptions}, \gls{se} yields no false positives and can be used in automated testing campaigns;
  on the other hand, program \emph{proving} with \gls{se} requires expensive auxiliary specifications~\cite{baumann.beckert.ea-12}. 
  
  So, what are these ``ideal assumptions?'' 
  First, \gls{se} is a \emph{whitebox} technique:
  It requires access to the source code of the tested program, interpreting its structure to drive the search for bugs.
  While this permits better results, implementing a symbolic executor requires a precise understanding of the semantics of the targeted language.
  For example, a symbolic executor for Python programs must respect Python's integer \emph{floor division} semantics:
  The division expression \code{x // y} corresponds to \code{floor(x / y)}, where ``\code{/}'' is division on reals. 
  Other interpretations can result in \textit{false positives} (spurious bugs) or \textit{false negative} (missed bugs).
  In the program
  \begin{tightcenter}
    \mintinline{python}{if x // y == -2: fail()}
  \end{tightcenter}
  a wrong interpretation of integer division might lead a symbolic bug finder to believe that $x = -5,\, y = 2$ triggers \code{fail()}.
  However, in Python, the result of \code{-5 // 2} is \(-3\), while in other languages such as Java it evaluates to \(-2\).
  
  Second, \gls{se} discards unreachable erroneous states by checking whether their path conditions are satisfiable.
  In our examples from above, it must decide if ``\mintinline{python}{x == hash((y, z))}'' or ``\code{x // y == -2}'' are consistent with the previously collected constraints.
  To that end, the \gls{se} engines use \emph{constraint solvers} to determine whether a set of constraints is satisfiable (\emph{SAT}) or unsatisfiable (\emph{UNSAT}).
  However, instead of returning \emph{SAT} or \emph{UNSAT}, solvers frequently \emph{time out} providing no answers; in that case, blindly reporting a bug could result in a \emph{false positive.}

  Consequently, an \gls{se} engine is more difficult to implement correctly than a fuzzer.
  While the semantics and correctness of \gls{se} have been studied in different contexts~\cite{steinhoefel-20,boer.bonsangue-19,lucanu.rusu.ea-17,kneuper-91}, the correctness of symbolic \emph{bug finders} remains an open research topic. 
  In particular, to our knowledge, there exists no rigorous definition of what constitutes a trustworthy symbolic bug finder.

  The main correctness criterion for testing tools is the ``soundness of bugs''~\cite{godefroid-05}:
  If the tool answers \emph{Is there a bug in my program?} with \emph{yes}, there must exist a bug witness (an actionable input).
  We call this property soundness for testing (\emph{T-soundness}).
  \emph{Fuzz testers} (short \emph{fuzzers}) generate inputs that they feed to the program under test and report a bug if the program crashes.
  Such tools are T-sound by construction.
  However, T-soundness alone is an insufficient quality measure for bug finders.
  Indeed, a tester never reporting \emph{any} bug is trivially T-sound.
  Thus, the \emph{converse} of soundness is an additionally desirable property:
  If a bug witness exists, a tester should eventually answer the question \emph{Is there a bug in my program?} with \emph{yes}.
  This property, which we name completeness for testing (\emph{T-completeness}), is generally out of reach for fuzzers.
  The best option is to estimate the \emph{residual risk} that an ongoing fuzzing campaign would discover an error if it only ran for a bit longer~\cite{boehme.liyange.ea-21}.
 
  Contrary, building a T-sound and T-complete symbolic executor \emph{is} possible.
  In this paper, we propose a general approach to building provably sound and complete bug finders based on symbolic execution.
  We implemented a prototypical bug finder based on this approach in the Coq proof assistant together with an end-to-end proof of its correctness.
  To the best of our knowledge, ours is the first endeavor to explicitly specify the \emph{meaning of correctness in the context of symbolic bug finders} \rev{in a proof assistant} and to use \emph{formal verification techniques} to prove the correctness of a bug finder.

  \rev{
  \paragraph*{Research Questions}

  Precisely, we answer the following main questions:
  \begin{enumerate*}[ref=(\arabic*)]
    \item \label[question]{q1}\emph{What constitutes a reliable testing tool?}, and
    \item \label[question]{q2}\emph{How can we engineer such a tool?}
  \end{enumerate*}
  These questions \emph{cannot be answered experimentally.}
  Following the terminology of Shaw~\cite{shaw-03}, they belong to the categories ``Method for Evaluation'' (\cref{q1}) and ``Method/Means of Development'' (\cref{q2}) of software engineering research questions.
  Our answer to \cref{q1} is the precise definition of T-soundness and T-completeness.
  Addressing \cref{q2}, we explain the characteristics of a T-sound and T-complete bug finder and expound a methodology to implement such a tool.
  Following this methodology, we implemented SE-based testing tools in Coq and Python.
  For the Coq prototype, we derived a formal, mechanized correctness proof.
  }
  

  \paragraph*{Related Work} 

  Kapus and Cadar~\cite{kapus.cadar-17} combine compiler testing techniques and different oracles to test three symbolic bug finders.
  They reported 20 distinct bugs ``exposing a variety of important errors''---e.g., related to division.
  This is the only work on the correctness of (symbolic) bug finders \rev{(\cref{q1})} we are aware of.
  \rev{Other work mainly studies how \gls{se} can be \emph{applied} for automated testing.}
  The strongest oracle used by Kapus and Cadar, the ``output oracle,'' compares the results of a concrete execution and a symbolic one following the same path.
  This oracle only works if the initial path condition \rev{uniquely describes the inputs used for concrete execution,} but is, at the same time, ``obfuscated'' to prevent the executor from dropping to ``concrete mode.''
  When testing \gls{se} engines with less constrained inputs, such that multiple paths are explored, Kapus and Cadar use coarser oracles.
  In contrast, we developed a \emph{verified} symbolic bug finder based on carefully designed semantic foundations \rev{in response to \cref{q2}}.
  Likely, Kapus and Cadar would have found more bugs when investing even more work on test cases and oracles; our bug finder is guaranteed to be bug-free.

  In formal program verification, it is much more common to address the soundness of verification tools and their underlying theoretical underpinnings.
  The goal of formal verification is not to show the presence but \emph{prove the absence} of bugs.
  Consequently, the main correctness criterion is the soundness of proofs (P-soundness):
  If a verifier answers \textit{Is my program safe to execute?} with \textit{yes}, then there exists a proof of the program's correctness.
  As with T-completeness, P-completeness is the converse of P-soundness.
  Since checking a complex program proof is arguably more complex then checking whether an input causes a bug, the formal verification community put significant effort into justifying their tools and foundations.
  For example, \gls{ai}~\cite{CousotCousot77-1} is a mathematical framework for designing and implementing correct static analysis.
  Following the \gls{ai} framework, Jourdan et al.~\cite{jourdan.laporte.ea-15} proved the soundness of an abstract interpreter for C programs in the Coq proof assistant.
  This project is strongly related to the development of CompCert~\cite{leroy-09}, a C compiler also proven correct in Coq.
  Coq itself is based on an established meta logic; its kernel was proven correct~\cite{sozequ.boulier.ea-20}.

  Closing the circle to symbolic execution, there have been various efforts to show \emph{P-soundness} of \gls{se}-based verifiers~\cite{ahrendt.roth.ea-05,keuchel.huyghebaert.ea-22,vogels.jacobs.ea-15}.
  Furthermore, the semantics of \gls{se} was addressed in the past~\cite{steinhoefel-20,boer.bonsangue-19,lucanu.rusu.ea-17,kneuper-91} in a more general context.
  The symbolic semantics used in this paper is inspired by the semantics proposed by de Boer and Bonsangue~\cite{boer.bonsangue-19}.
  Building on these foundations, we define soundness and completeness in the context of symbolic bug finding and provide a formally verified implementation.

  \paragraph*{Contributions}
  
  Following the CompCert approach~\cite{Robert_Leroy_CPP2012,10.1007/978-3-642-38856-9_18} to build verified program analyzers in a proof assistant, we develop a prototype of a formally verified symbolic bug finder in Coq.
  We formalize the concrete operational semantics of a small imperative programming language and derive a symbolic operational semantics inspired by de Boer and Bonsangue~\cite{boer.bonsangue-19} (\cref{sec:concrete_and_symbolic_semantics}). 
  Then, we extensively explore the relations between concrete and symbolic semantics by proving that \gls{se} is a sound and complete method for reachability analysis and various bug-finding tasks (\cref{sec:finding_bugs}).
  Building on these theoretical results, we derive a T-sound and T-complete symbolic bug finder (\cref{sec:trustworthy_bug_finder}).
  In the course of this, we formally address several practical questions such as \emph{path selection} and \emph{search space pruning} using constraint solvers and their influence on \emph{exhaustiveness} (i.e., T-completeness) and \emph{precision} (i.e., T-soundness) of \gls{se}.
  All implementations, definitions, theorems, and proofs in this paper are realized in Coq.
  Thus, if our \emph{specifications} convince you, you can \emph{blindly trust the proofs.}
  We call our prototypical symbolic bug finder \wise{} (``What is Symbolic Execution'').
  To demonstrate the practicability of our \gls{se} engine's design principles, we (1) extract an executable OCaml implementation fully automatically from the Coq code and (2) manually transform \wise{} to \pywise{}, a \wise{} implementation written in Python (\cref{sec:case_studies}).
  Both extracted bug finders can be invoked on the command line. 
  To provide hands-on evidence that \wise{} works---additionally to our formal proofs---we apply it to find bugs inserted into numeric algorithms. 

  \section{Concrete and Symbolic Semantics}%
  \label{sec:concrete_and_symbolic_semantics}

  We study symbolic execution along the example of IMP, a small imperative programming language with loops and integer arithmetic.
  In this section, we introduce the IMP language (\cref{sub:imp}) and describe its operational semantics (\cref{sub:imp_concrete_semantics}).
  Afterward, we equip IMP with a symbolic operational semantics following de Boer and Bonsangue~\cite{boer.bonsangue-19} (\cref{sub:imp_symbolic_semantics}).
  The symbolic semantics will be the formal basis to justify the correctness of our \gls{se} engine.

  \paragraph*{Notation}

  We use the following notational conventions to clarify the domains of identifiers in our formalizations.
  \begin{itemize}
    \item \(\prog\), \(\progb\), \dots{} for \emph{programs}
    \item \(\cstore\) for \emph{concrete environments}, \(\sstore\) for \emph{symbolic stores}
    \item \(\pathcond\) for \emph{path conditions} or \emph{\glstext{ltl} formulas}
    \item \(\cstate\) for \emph{concrete states} \(\langle{}V,p\rangle\), \(\sstate\) for \emph{symbolic states} \(\langle\pathcond,\sstore,\prog\rangle\)
    \item \(\stream\) for \emph{streams} (e.g., of symbolic states)
  \end{itemize}
  
  \subsection{A Small Imperative Programming Language}
  \label{sub:imp}

  IMP is an imperative programming language with assignments, while loops, and conditionals.
  \Cref{fig:imp_syntax} contains a grammar of IMP's concrete syntax; \(\mathit{VAR}\) denotes a variable identifier and \(\mathit{INT}\) an integer.
  We integrate a \syn{fail} statement signaling an error.
  Using \syn{fail}, one can instrument source code to, e.g., mark supposedly unreachable sections or model run-time assertions.
  IMP does not natively support multiplication and division;
  these operators can be implemented in terms of addition and subtraction.
  
  \begin{figure}[tb]
    \begin{bnfgrammar}
      stmt : \text{Statements} ::= \syn{skip}
      | \syn{fail}
      | VAR \syn{=} aexpr
      | stmt \syn{;} stmt
      | \syn{while} bexpr \syn{do} stmt \syn{od}
      | \syn{if} bexpr \syn{then} stmt \syn{else} stmt \syn{fi}
      ;;
      bexpr : Expressions ::= \syn{true} || \syn{false}
      | bexpr ( \syn{and} \(\vert\) \syn{or} ) bexpr
      | \syn{not} bexpr
      | aexpr ( \syn{==} \(\vert\) \syn{<=} \(\vert\) \syn{<} \(\vert\) \syn{>=} \(\vert\) \syn{>} ) aexpr
      ;;
      aexpr ::= INT || VAR
      | aexpr ( \syn{+} \(\vert\) \syn{-} ) aexpr 
    \end{bnfgrammar}
    \Description{Syntax of IMP}%
    \caption{Syntax of IMP}%
    \label{fig:imp_syntax}
  \end{figure}

  \begin{figure}[htb]
    \begin{mathpar}
      \inferrule{}{\st{V}{x \ \syn{=} \ e} \to \st{V[x := \den{e}_V]}{\syn{skip}}}\\
      \inferrule{\den{b}_V = \syn{true}}{\st{V}{\syn{if} \ b \ \syn{then} \ s_1 \ \syn{else} \ s_2\ \syn{fi}} \to \st{V}{s_1}}\\
      \inferrule{\den{b}_V = \syn{false}}{\st{V}{\syn{if} \ b \ \syn{then} \ s_1 \ \syn{else} \ s_2\ \syn{fi}} \to \st{V}{s_2}}\\
      \inferrule{\st{V_1}{s_1} \to \st{V_2}{s_2}}{\st{V_1}{s_1 \ \syn{;} \ s_3} \to \st{V_2}{s_2 \ \syn{;} \ s_3}}\\
      \inferrule{}{\st{V}{\syn{skip} \ \syn{;} \ s} \to \st{V}{s}}\\
      \inferrule{\den{b}_V = \syn{true}}{\st{V}{\syn{while} \ b \ \syn{do} \ s\ \syn{od}} \to \st{V}{s \ \syn{;} \ \syn{while} \ b \ \syn{do} \ s\ \syn{od}}}\\
      \inferrule{\den{b}_V = \syn{false}}{\st{V}{\syn{while} \ b \ \syn{do} \ s\ \syn{od}} \to \st{V}{\syn{skip}}}\\
      \den{x}_V := V(x), \ x\in \textrm{Vars}\\
      \den{\syn{true}}_V := \syn{true} \qquad 
      \den{\syn{false}}_V := \syn{false}\\
      \den{e_1 \ \syn{+} \ e_2}_V := \den{e_1}_V + \den{e_2}_V \qquad \den{b_1 \ \syn{and} \ b_2}_V := \den{b_1}_V \wedge \den{b_2}_V\\
      \ldots
    \end{mathpar}
    \Description{Operational Semantics of IMP}%
    \caption{Operational Semantics of IMP}%
    \label{fig:concrete}
  \end{figure}

  \subsection{Concrete Operational Semantics}%
  \label{sub:imp_concrete_semantics}

  IMP programs are executed in an \emph{environment} $\cstore : \textrm{Vars} \to \mathbb{Z}$ representing the values assigned to each program variable.
  Execution states are pairs $\langle V, p\rangle$ of an environment and a program that remains to be executed.
  The small-step operational semantics for IMP (\cref{fig:concrete}) is a relation between execution states.
  We use the notation $\langle V, p\rangle \to \langle V', p'\rangle$ to state that executing program $p$ in environment $V$ for one step produces a new environment $V'$ together with a program $p'$ that remains to be executed.
  The relation ``\(\to^*\)'' is the reflexive, transitive closure of ``\(\to\).''
  We use the notation $\llbracket e \rrbracket_V$ to denote the value of the expression $e$ in environment $V$.
  The operational semantics is presented in \cref{fig:concrete}.

  \subsection{Symbolic Operational Semantics}%
  \label{sub:imp_symbolic_semantics}

  Informally, a program is symbolically executed by using symbolic placeholders instead of concrete values as inputs.
  When a statement with different possible outcomes is reached (in IMP, this can be an \syn{if} statement or a loop), \emph{both} possibilities are independently explored in different execution branches.
  Symbolic execution thus produces in a \emph{tree} representing all possible execution paths.
  Each tree node is labeled with a formula (called the \emph{path condition}), which characterizes the condition on the program inputs that needs to be met for the program to reach the associated state.

  More formally, we regard \gls{se} as a type of program execution respecting a \emph{symbolic} operational semantics ``$\to_{\textrm{sym}}$'' (\cref{fig:symbolic}).
  In \gls{se}, symbolic \emph{stores} \(S : \textrm{Vars} \to \textit{expr}\) take the role of environments.
  Contrary to the latter, symbolic stores map variables to symbolic expressions and not numeric values.
  The dual to concrete execution states are \emph{symbolic states}, triples $\langle \varphi, S, p\rangle$ of a formula \(\varphi\) (the \emph{path condition}), a \emph{symbolic store} \(S\), and a program to be executed \(p\) (the \emph{program counter}).
  We extend the concrete evaluation of an expression $e$ in a store $V$ to a symbolic evaluation $\den{e}_S$.
  The result of \(\den{e}_S\) is again a symbolic expression obtained by replacing variables in \(e\) with their associated expressions in \(S\), if any.

  \begin{figure}[tb]
    \begin{mathpar}
      \inferrule{}{\sst{\varphi}{S}{x \ \syn{=} \ e} \to_{\textrm{sym}} \sst{\varphi}{S[x := \den{e}_S]}{\syn{skip}}}\\
      \inferrule{}{\sst{\varphi}{S}{\syn{if} \ b \ \syn{then} \ s_1 \ \syn{else} \ s_2\ \syn{fi}} \to_{\textrm{sym}} \sst{\varphi \ \syn{and} \ \den{b}_S}{S}{s_1}}\\
      \inferrule{}{\sst{\varphi}{S}{\syn{if} \ b \ \syn{then} \ s_1 \ \syn{else} \ s_2\ \syn{fi}} \to_{\textrm{sym}} \sst{\varphi \ \syn{and} \ (\syn{not} \ \den{b}_S)}{S}{s_2}}\\
      \inferrule{\sst{\varphi_1}{S_1}{s_1} \to \sst{\varphi_2}{S_2}{s_2}}{\sst{\varphi_1}{S_1}{s_1 \ \syn{;} \ s_3} \to_{\textrm{sym}} \sst{\varphi_2}{S_2}{s_2 \ \syn{;} \ s_3}}\\
      \inferrule{}{\sst{\varphi}{S}{\syn{skip} \  \syn{;} \ s} \to_{\textrm{sym}} \sst{\varphi}{S}{s}}\\
      \inferrule{}{\sst{\varphi}{S}{\syn{while} \ b \ \syn{do} \ s\ \syn{od}} \to_{\textrm{sym}} \sst{\varphi \ \syn{and} \ \den{b}_S}{S}{s \ \syn{;} \ \syn{while} \ b \ \syn{do} \ s\ \syn{od}}}\\
      \qquad\inferrule{}{\sst{\varphi}{S}{\syn{while} \ b \ \syn{do} \ s\ \syn{od}} \to_{\textrm{sym}} \sst{\varphi \ \syn{and} \ (\syn{not} \ \den{b}_S)}{S}{\syn{skip}}}
    \end{mathpar}
    \Description{Symbolic Semantics of IMP}%
    \caption{Symbolic Semantics of IMP}%
    \label{fig:symbolic}
  \end{figure}

  The symbolic execution of a program $p$ usually starts in the initial state $\sst{\syn{true}}{id}{p}$, where $id := x \mapsto x$ is an ``empty'' symbolic store assigning to each program variable a symbol of the same name.
  The initial path condition $\syn{true}$ indicates that we make no assumptions about the program inputs.
  We write $\sst{\syn{true}}{id}{p} \to^*_\textrm{sym} \sst{\varphi}{S}{p'}$ to denote that the symbolic state $\sst{\varphi}{S}{p'}$ can be reached by symbolic execution of $p$.
  \Cref{fig:example_symbolic_execution} shows a symbolic execution tree for the execution of a conditional statement in the default initial state.

  \begin{figure}[htb]
    \centerline{\begin{tikzpicture}
      \node[shape=rectangle,draw=black] (A) at (0,0) {$\sst{\syn{true}}{x \mapsto x}{\syn{if} \ \texttt{x} < 0 \ \syn{then} \ \syn{fail} \ \syn{else} \ \texttt{x} \ \syn{=} \ \texttt{x} \ \syn{-} \ 1\ \syn{fi}}$};
      \node[shape=rectangle,draw=black] (B) at (-2,-1) {$\sst{\texttt{x} \ \syn{<} \ 0}{x \mapsto x}{\syn{fail}}$};
      \node[shape=rectangle,draw=black] (C) at (2,-1) {$\sst{\texttt{x} \ \syn{>=} \ 0}{x \mapsto x}{\texttt{x} \ \syn{=} \ \texttt{x} \ \syn{-} \ 1}$};
      \node[shape=rectangle,draw=black] (D) at (2,-2) {$\sst{\texttt{x} \ \syn{>=} \ 0}{x \mapsto x - 1}{\syn{skip}}$};

      \path [->] (A) edge node[left] {\smaller $\textrm{sym}$} (B);
      \path [->] (A) edge node[right] {\smaller $\textrm{sym}$} (C);
      \path [->] (C) edge node[right] {\smaller $\textrm{sym}$} (D);
    \end{tikzpicture}}
    \Description{Example Symbolic Execution Tree}%
    \caption{Example Symbolic Execution Tree}%
    \label{fig:example_symbolic_execution}
  \end{figure}

  \section{Finding Bugs Symbolically}%
  \label{sec:finding_bugs}
  
  A symbolic executor \textit{simulates} many concrete executions of a program simultaneously.
  From this, we can obtain a bug finder by reporting an error whenever one simulated execution reaches an error state.
  Every reported bug must come with an \textit{witness}.
  In fuzzing, this is a bug-triggering input; in the case of \gls{se}, it is frequently a more abstract witness in the form of a boolean formula describing \emph{all} inputs triggering the bug.
  Additionally, we would like to ensure that any bug in the program will eventually be detected during \gls{se}.
  In this section, formally connect bug finding and symbolic execution.
  In particular, we lift the symbolic semantics from \cref{sub:imp_symbolic_semantics} to a sound and complete method to find bugs.

  \subsection{Sound \& Complete Reachability}

  Our first step toward a symbolic bug finder is to connect symbolic and concrete executions.
  The glue we use to establish this connection is \emph{reachability.}
  A correct symbolic executor should \emph{only} reach symbolic states representing reachable concrete states.
  Conversely, we would like it to reach \emph{any} reachable concrete state in the form of a symbolic state representing it.
  
  We write $\textit{Reach}(\prog, \cstate)$ and $\textit{Reach}_\textrm{sym}(\prog, \cstate)$ to express that concrete state $\cstate$ can be reached by concrete and symbolic execution of the program $\prog$.
  Then, we express soundness and completeness concerning reachability as a simple equivalence:

  \begin{theorem}[Soundness and Completeness for Reachability]
    \hspace*{10pt}$\textit{Reach}_\textrm{sym}(\prog, \cstate) \Leftrightarrow \textit{Reach}(\prog, \cstate)$
  \end{theorem}

  $\textit{Reach}(\prog, \cstate)$ can be expressed using the concrete operational semantics.
  A state $\cstate$ is reachable for a program $\prog$ if executing $\prog$ in some initial store $\cstore_0$ eventually leads to $\cstate$.

  \begin{definition}[Concrete Reachability]
    \hspace*{10pt}$\textit{Reach}(\prog, \cstate) \triangleq \exists \cstore_0, \st{\cstore_0}{\prog} \to^* \cstate$
  \end{definition}

  Symbolic states reached during \gls{se} represent \emph{many} concrete states.
  To make \(\reachsym(\prog,\cstate)\) precise, we relate them to their ``concretizations.'' 
  The concretization operation $\circ$ turns a symbolic store $\sstore$ into a concrete store $\cstore'=\cstore \circ \sstore$ by replacing all free variables in $\sstore$ by their associated values in $\cstore$.
  For example, if \(\sstore=\{x\mapsto{}y,y\mapsto{}y\}\) and \(\cstore=\{x\mapsto{}1,y\mapsto{}2\}\), we obtain \(\cstore'=\{x\mapsto{}2,y\mapsto{}2\}\). 

  \begin{definition}[][comp]
    \hspace*{10pt}$(V \circ S)(x) := \den{S(x)}_V$
  \end{definition}

  We lift the \emph{store} concretization ``\(\circ\)'' to the concretization of symbolic \emph{states}.
  To express that a symbolic state \(\sstate\) represents a concrete state \(\cstate\) via an initial store \(\cstore_0\), we write \(\cstate\stateconcr{\cstore_0}\sstate\).
  Any concretization \(\cstate\) of \(\sstate\) must have the \emph{same program counter} as \(\cstate\).
  Furthermore, the initial store \(\cstore_0\) must \emph{satisfy the path condition} of \(\sstate\).
  Finally, \(\cstate\)'s environment must be a concretization of \(\sstate\)'s store via \(\cstore_0\).

  \begin{definition}[][sim]
    \hspace*{10pt}$\st{V}{p} \simeq_{V_0} \sst{\varphi}{S}{p'} \triangleq p = p' \wedge V = V_0 \circ S \wedge \den{\varphi}_{V_0} = \syn{true}$
  \end{definition}

  Now, all \emph{concretizations of reachable symbolic states} are symbolically reachable.
  Formally:

  \begin{definition}
    \hspace*{10pt}$\reachsym(\prog, \cstate) \triangleq \exists \cstore_0, \exists \sstate, \sst{\syn{true}}{id}{p} \to^*_\textrm{sym} \sstate \wedge \cstate \stateconcr{\cstore_0} \sstate$
  \end{definition}

  In the remainder of this section, we describe the proofs of both directions of Theorem~2.1 for the interested reader.
  Following the proofs, \cref{sub:sound_complete_bug_finding} addresses the detection of bugs, which involves making precise what constitutes an ``error state.''


  \section*{\centering \small Proof of Soundness}

  Subsequently, we derive the proof of ``sound reachability,'' i.e., that all states reachable by \gls{se} are also reachable by concretely executing a program.
  This proof requires some intermediate lemmas.
  \cref{coq:comp_update} is a compositionality result on the interplay of concretization and symbolic store updates.

  \begin{lemma}[][comp_update]
    \hspace*{10pt}$V \circ (S [x := \den{e}_S]) = (V \circ S)[x := \den{e}_{V \circ S}]$
  \end{lemma}
  \begin{proof}
    By structural induction on $e$.
  \end{proof}

  During \gls{se}, path conditions get only ever \emph{more specific} by accumulating addition constraints.
  This is asserted in \cref{coq:sym_steps_path}.
  We note $\varphi \vDash \varphi' \triangleq \forall V, \den{\varphi}_V = \syn{true} \Rightarrow \den{\varphi'}_V = \syn{true}$.

  \begin{lemma}[][sym_steps_path]
    \hspace*{10pt}$\langle \varphi
    , S, p \rangle \rightarrow^*_{\textrm{sym}} \langle \varphi', S', p' \rangle \Rightarrow
    \ \varphi' \vDash \varphi$
  \end{lemma}
  \begin{proof}
    By induction on the length of the $\to^*_\textrm{sym}$ derivation and by induction
    on the last $\to_\textrm{sym}$ step.
  \end{proof}

  Whenever the path condition after a symbolic step is satisfiable, we can make a corresponding concrete step.

  \begin{lemma}[][sym_step_step]
    \hspace*{10pt}$\langle \varphi
    , S, p \rangle \rightarrow_{\textrm{sym}} \langle \varphi', S', p' \rangle \wedge V \vDash \varphi' \Rightarrow \langle V \circ S, p \rangle \rightarrow \langle V \circ S', p' \rangle$
  \end{lemma}
  \begin{proof} By induction over the symbolic derivation and by composition (\ref{coq:comp_update}) in the case of variable assignments $x \ \syn{=} \ e$.
  \end{proof}

  We extend \cref{coq:sym_step_step} to the transitive reflexive closure by applying monotonicity (\cref{coq:sym_steps_path}).

  \begin{lemma}[][sym_steps_steps]
    \hspace*{10pt}$\langle \varphi
    , S, p \rangle \rightarrow^*_{\textrm{sym}} \langle \varphi', S', p' \rangle \wedge V \vDash \varphi' \Rightarrow \langle V \circ S, p \rangle \rightarrow^* \langle V \circ S', p' \rangle$
  \end{lemma}
  \begin{proof} By iterating \cref{coq:sym_step_step} and by \cref{coq:sym_steps_path}.
  \end{proof}

  Sound reachability follows from the (more general) \cref{coq:sym_steps_steps}.

  \begin{theorem}[Sound Reachability][Reach_sound]
    \hspace*{10pt}$\textit{Reach}_\textrm{sym}(p, \sigma) \Rightarrow \textit{Reach}(p, \sigma)$
  \end{theorem}
  \begin{proof}
    By immediate application of \cref{coq:sym_steps_steps} and the fact that $\st{V_0}{p} \simeq_{V_0} \sst{\syn{true}}{id}{p}$.
  \end{proof}

  \par\noindent\rule{\linewidth}{.8pt}
  
  \section*{\centering \small Proof of Completeness}
  
  We address the converse direction:
  All concretely reachable states are symbolically reachable.
  The idea is to show that for any execution step that can be taken according to the concrete semantics, there exists a symbolic step to simulate it.

  \begin{lemma}[][step_simulation_diagram]
    \hspace*{10pt}$(\cstate \simeq_{\cstore_0} \sstate \wedge \cstate \rightarrow \cstate') \Rightarrow \exists \sstate', \sstate \rightarrow_{\textrm{sym}} \sstate' \wedge \cstate' \simeq_{\cstore_0} \sstate'$
  \end{lemma}
  \begin{proof}
    By induction on the derivation $\cstate \rightarrow \cstate'$ and selecting the appropriate symbolic rule for each case.
  \end{proof}

  We extend \cref{coq:step_simulation_diagram} to the reflexive transitive closure.

  \begin{lemma}[][steps_simulation_diagram]
    \hspace*{10pt}$(\cstate \simeq_{\cstore_0} \sstate \wedge \cstate \rightarrow^* \cstate') \Rightarrow \exists \sstate', \sstate \rightarrow^*_{\textrm{sym}} \sstate' \wedge \cstate' \simeq_{\cstore_0} \sstate'$
  \end{lemma}
  \begin{proof}
    By iterating \cref{coq:step_simulation_diagram}.
  \end{proof}

  The completeness of $\textit{Reach}_\textrm{sym}$ follows directly from \cref{coq:steps_simulation_diagram}.


  \begin{theorem}[Complete Reachability][Reach_complete]
    \hspace*{10pt}$\textit{Reach}(p, \sigma) \Rightarrow \textit{Reach}_\textrm{sym}(p, \sigma)$
  \end{theorem}
  \begin{proof}
    It is easy to see that $\st{V_0}{p} \simeq_{V_0} \sst{\syn{true}}{id}{p}$. The result immediately follows from this observation and \cref{coq:steps_simulation_diagram}.
  \end{proof}

  \subsection{Sound \& Complete Bug Finding}
  \label{sub:sound_complete_bug_finding}

  In the previous section, we proved that the states reachable by \gls{se} are the same as those reachable by concrete execution.
  Building on that, we establish that \gls{se} can also be used as a sound and complete bug-finding method.
  First, however, we must address the question \emph{What does it mean to find bugs?}
  We present different views on this question and prove that \gls{se} answers each case.

  \paragraph*{Detecting Buggy Programs}

  A natural view on bug finding is to ask---and answer---the question \emph{Does my program contain a bug?}
  We introduce a relation $\textit{HasBug}(\prog)$ characterizing buggy programs.
  This predicate is central to our definition of T-soundness and T-completeness (\cref{sub:tsound_tcomplete_bug_finder}).
  A \emph{state} \(\sigma=\st{\cstore}{\prog}\) ``has a bug'' if $\prog$ is \emph{stuck}, i.e., has no semantic successor state.
  However (in the case of IMP), this would also hold if \(\prog\) was \(\syn{skip}\), although we were more thinking of \(\syn{fail}\).
  We explicitly exclude the \syn{skip} case in the following definitions of ``not stuck'' (i.e., ``can progress'') and ``stuck.''

  \setCoqFilename{WiSE.lang.imp}
  \begin{definition}[Progress][progress]
    \hspace*{10pt}$\textit{Progress}(\st{V}{p}) \triangleq p = \syn{skip} \vee \exists \sigma, \st{V}{p} \to \sigma$
  \end{definition}

  \begin{definition}[Stuck States][Stuck]
    \hspace*{10pt}$\textit{Stuck}(\sigma) \triangleq \neg \textit{Progress}(\sigma)$
  \end{definition}

  A buggy program has a reachable, stuck state.

  \begin{definition}[Having a Bug][HasBug]
    \hspace*{10pt}$\textrm{HasBug}(p) \triangleq \exists \sigma, \textit{Reach}(p, \sigma) \wedge \textit{Stuck}(\sigma)$
  \end{definition}

  As the question \emph{Does my program contain a bug?} can be reduced to reachability, we can answer it symbolically.
  
  \setCoqFilename{WiSE.symbolic.symex}
  \begin{definition}[][HasBug]
    \hspace*{10pt}$\textrm{HasBug}_\textrm{sym}(p) \triangleq \exists \sigma, \textit{Reach}_\textrm{sym}(p, \sigma) \wedge \textit{Stuck}(\sigma)$
  \end{definition}

  The equivalence of \(\hasBugSym\) and \(\hasBug\) follows immediately from the equivalence of \(\reachsym\) and \(\reach\).

  \begin{theorem}[Sound and Complete Bug Detection (1)][HasBug_correct]
    \hspace*{10pt}$\textit{HasBug}_\textrm{sym}(p) \Leftrightarrow \textit{HasBug}(p)$
  \end{theorem}
  \begin{proof}
    By soundness of $\textit{Reach}_\textrm{sym}$ in the ``$\Rightarrow$'' direction and by completeness of $\textit{Reach}_\textrm{sym}$ in the ``$\Leftarrow$'' direction.
  \end{proof} 

  \paragraph*{Finding Erroneous Execution States}

  $\textit{HasBug}(p)$ asserts that \emph{some} bug exists.
  Alternatively, we can specifically ask whether a \emph{given} state is erroneous and whether \gls{se} can find it. 
  The predicate $\textit{IsBug}(p, \sigma)$ asserts that $\sigma$ is an error state of $p$.
  In \cref{sub:rel_sound_compl_bug_finder}, we derive a refined ``relative'' version of T-soundness and T-completeness for nonterminating symbolic executors using $\textit{IsBug}(p, \sigma)$.

  \setCoqFilename{WiSE.lang.imp}
  \begin{definition}[][IsBug]
    \hspace*{10pt}$\textit{IsBug}(p, \sigma) \triangleq \textit{Reach}(p, \sigma) \wedge \textit{Stuck}(\sigma)$
  \end{definition}

  Similarly to above, we derive a \emph{symbolic} predicate $\textit{IsBug}_\textrm{sym}(p, \sigma)$.
  
  \setCoqFilename{WiSE.symbolic.symex}
  \begin{theorem}[Sound and Complete Bug Detection (2)][IsBug_correct]
    \hspace*{10pt}$\textit{IsBug}_\textrm{sym}(p, \sigma) \Leftrightarrow \textit{IsBug}(p, \sigma)$
  \end{theorem}

  \paragraph*{Finding Erroneous Inputs}

  Our final view is to ask for ``bad,'' i.e., bug-triggering, \emph{inputs}.
  The predicate $\textit{BadInput}(p, V_0)$ characterizes those (where the environment \(\cstore_0\) represents an input):

  \setCoqFilename{WiSE.lang.imp}
  \begin{definition}[][BadInput]
    \hspace*{10pt}$\textit{BadInput}(p, V_0) \triangleq \exists \sigma,\st{V_0}{p} \to^* \sigma \wedge \textit{Stuck}(\sigma)$
  \end{definition}

  The symbolic equivalent is more subtle since $\to^*_\textrm{sym}$ ranges over \emph{symbolic} states. We must concretize the reached symbolic state before checking whether it denotes a bug.

  \setCoqFilename{WiSE.symbolic.symex}
  \begin{definition}[][BadInput]
    \hspace*{10pt}$\textit{BadInput}_\textrm{sym}(p, V_0) \triangleq \exists \sigma \exists \hat{\sigma},\\\hspace*{20pt}\sst{\syn{true}}{S}{p} \to^*_\textrm{sym} \hat{\sigma} \wedge \sigma \simeq_{V_0} \hat{\sigma} \wedge \textit{Stuck}(\sigma)$
  \end{definition}

  The equivalence of $\textit{BadInput}_\textrm{sym}$ and $\textit{BadInput}$ follows from the definition of $\textit{Reach}_\textrm{sym}$ and the proof of \cref{coq:Reach_sound} and \cref{coq:Reach_complete}.

  \begin{theorem}[Sound and Complete Bug Detection (3)][BadInput_correct]
    \hspace*{10pt}$\textit{BadInputs}_\textrm{sym}(p, V_0) \Leftrightarrow \textit{BadInputs}(p, V_0)$
  \end{theorem}

  This view bridges the gap between \gls{se} and fuzzing:
  What is needed to report a \emph{bug-triggering input} is (1) a bug \emph{oracle} (a symbolic equivalent of \(\mathit{Stuck}(\cstate)\)) and (2) a constraint solver extracting a solution \(\cstore_0\) from a path constraint.
  Subsequently, we describe how to build a testing tool based on our insights on \gls{se}, including a brief discussion of ``crash oracles.''
  Constraint solving is not in our scope.

  \section{A Trustworthy Symbolic Bug Finder}%
  \label{sec:trustworthy_bug_finder}

  \Cref{sec:finding_bugs} demonstrated the feasibility of finding bugs with \gls{se} based on the symbolic operational semantics.
  Now, we turn these insights into an executable correct symbolic bug finder.
  First, we discuss how to detect errors, resolve the nondeterminism of the symbolic semantics, and deal with diverging executions (\cref{sub:semantics_to_executor}).
  Second, we provide relevant implementation details of our \wise{} prototype (\cref{sec:implem}).
  Third, we introduce ``relative'' notions of soundness and completeness and relate them to the symbolic executor (\cref{sub:soundn-compln-se-engine}) and bug finder (\cref{sub:rel_sound_compl_bug_finder}).
  Finally, we bridge the (already tiny) gap toward a T-sound and T-complete bug finder (\cref{sub:tsound_tcomplete_bug_finder}).
  
  \subsection{From the Symbolic Semantics to an Executable Symbolic Executor}%
  \label{sub:semantics_to_executor}

  Our symbolic operational semantics is a solid basis for symbolic bug finding.
  Yet, its design leaves some crucial questions open that need to be addressed when developing an executable \gls{se}-based testing tool.
  These are:
  (1) The symbolic semantics is \emph{nondeterministic}. 
  Whenever an \syn{if} or \syn{while} statement is reached, \gls{se} can \textit{choose} which branch to execute next.
  (2) We need an effective way to detect error states.
  The previously regarded predicates \(\textit{HasBug}_\textrm{sym}\), $\textit{IsBug}_\textrm{sym}$, and $\textit{BadInput}_\textrm{sym}$ are adequate formalizations for bug detection but are not executable because they live on the \emph{semantic} level. 
  Finally, (3) the symbolic semantics does \emph{not terminate} for programs with loops, which we cannot solve in general, but must face nonetheless.

  
  \begin{figure}[htb]
    \centering
    \begin{gather*}
      \textrm{HasBug}_\textrm{sym}(p)\\
      \Leftrightarrow\\
      \exists \sigma, \textit{Reach}_\textrm{sym}(p, \sigma) \wedge \textit{Stuck}
      (\sigma)\\
      \Leftrightarrow\\
      \exists \sigma, (\exists V_0, \exists \hat{\sigma}, \sst{\syn{true}}{id}{p} \to^*_\textrm{sym} \hat{\sigma} \wedge \sigma \simeq_{V_0} \hat{\sigma}) \wedge \textit{Stuck}(\sigma)\\
      \Leftrightarrow\\
      \exists V_0, \exists \sst{\varphi}{S}{q},\\\sst{\syn{true}}{id}{p} \to^*_\textrm{sym} \sst{\varphi}{S}{q} \wedge \den{\varphi}_{V_0} = \syn{true} \wedge \textit{Stuck}(\st{V_0 \circ S}{q})\\
      \Leftrightarrow\\
      \exists \sst{\varphi}{S}{q}, \fbox{$\sst{\syn{true}}{id}{p} \to^*_\textrm{sym} \sst{\varphi}{S}{q}$}_{(1)} \wedge
      \\\fbox{$\exists V_0, \den{\varphi}_{V_0} = \syn{true}$}_{(2)} \wedge \fbox{$\textit{Stuck}(\st{V_0 \circ S}{q})$}_{(3)}
    \end{gather*}
    \Description{Decomposition of symbolic bug finding into checks for (1) reachability, (2) satisfiability, and (3) ``stuckness.''}%
    \caption{Decomposition of symbolic bug finding into checks for (1) reachability, (2) satisfiability, and (3) ``stuckness.''}
    \label{fig:hiddensat}
  \end{figure}

  \paragraph*{Resolving Nondeterminism}

  While in theory, we could spawn independent processes whenever \gls{se} performs a case distinction when facing an \syn{if} or \syn{while} statement, this cannot continue forever.
  The number of branches in a symbolic execution tree is exponential in the number of such branching points; eventually, we will \emph{have to} decide which branches to follow first.
  This process is known as \emph{path selection} in \gls{se}~\cite{baldoni.coppa.ea-18}.
  Usually, path selection is regarded as a ``heuristics,'' merely an artifact of efficiency considerations.
  However, the exploration order has solid consequences on the \emph{T-completeness} of the resulting bug finder.
  Early dynamic symbolic executors such as DART~\cite{godefroid.klarlund.ea-05} explored programs using a Depth-First Search (DFS)\@.
  This approach is \emph{incomplete} in general:
  Such executors likely get lost in ever deeper iterations of loops at the beginning of a program, neglecting any states that follow.
  Instead, we choose a Breadth-First Search (BFS) approach.
  In \cref{sub:soundn-compln-se-engine}, we prove the completeness of BFS-based path selection, which requires estimating where a successor state is added in the stream of states to explore.
  
  \paragraph*{Detecting Errors}

  Any bug finder depends on a \emph{bug oracle} to determine when a program state constitutes a bug. 
  Predicates such as $\textit{HasBug}_\textrm{sym}$ are non-executable \emph{semantic} notions.
  Moreover, they are defined using symbolic reachability, which requires checking if a path condition is \emph{satisfiable} (see \cref{fig:hiddensat}).
  Fuzzers commonly classify program runs as erroneous if they end with a \emph{crash} of the program.
  We use a similar approach:
  We classify a program as buggy or crashing if it is \emph{stuck.}
  In our case, it can be shown that a state is stuck if and only the next statement to be executed is \syn{fail}.
  This simple solution is already quite powerful.
  For example, to extend IMP with failing \emph{expressions,} one can use \syn{fail} to model the exceptional cases.
  Consider integer division:
  For any expression \code{x / y}, we add a preceding statement \code{\syn{if} y == 0 \syn{then fail else} d = div(x,y)}, replacing \code{x / y} with the fresh variable \code{d}.
  The function \code{div} would implement division using repeated subtraction, \emph{assuming} that \code{y} is nonzero.

  \paragraph*{Addressing Termination}

  As a method simulating concrete executions, \gls{se} does not terminate for most looping or recursive programs.
  Derivations in our symbolic semantics (\cref{fig:symbolic}) are infinite for \emph{any} program with a loop.
  In \cref{sub:tsound_tcomplete_bug_finder}, we integrate pruning of unsatisfiable states, which solves that issue for loops with concrete or unsatisfiable guards.
  Even so, loops with underconstrained symbolic guards induce an infinite amount of reachable successor states.

  There are two main options for dealing with nontermination in \gls{se}.
  The first option is to impose an upper bound on the number of times a loop is unwound.
  This approach has the advantage that it is easy to implement.
  However, it prevents us from reasoning about the \emph{T-completeness} of a symbolic executor, since the chosen bound might be too small to detect sufficiently deep bugs.
  The approach we chose instead is to implement the symbolic executor as a generator of an \emph{infinite stream} of states.
  A user of a stream-based executor can only ever inspect finitely many states from the stream, which corresponds to bounded execution;
  yet, we \emph{can obtain a completeness result.}
  Our OCaml and Python implementations (\cref{sec:case_studies}) retrieve elements from the stream up to a user-specified bound.


  Thus, we view automated bug finders as producers of lazily evaluated infinite sequences of status messages reporting on the progress of the bug finding process.
  There are various ways of implementing infinite streams depending on the capabilities of the programming language a symbolic executor is implemented in.
  In proof systems based on purely functional languages such as Coq, infinite sequences can be modeled by \emph{co-induction}~\cite{10.1007/3-540-61780-9_67,Gimnez2005ATO}.
  In languages with lazy evaluation such as Haskell, such infinite sequences can be implemented using standard \emph{lists}.
  In the case of eagerly evaluated functional languages such as OCaml, the implementation can resort to \emph{explicit suspensions}.
  Otherwise, we can model streams with \emph{generators} \textit{à la} Python.

  In \cref{sub:soundn-compln-se-engine}, we show that the stream-based approach permits elegant specifications of correctness based on temporal logic.

  \subsection{Functional Implementation}
  \label{sec:implem}

  After many dry discussions, we provide some implementation details of our proven-correct symbolic executor.
  As discussed in \cref{sub:semantics_to_executor}, the executor uses a breadth-first goal selection and produces a ``lazy'' infinite stream of status messages.
  We use an OCaml-flavored syntax in all code snippets; we found it easy to translate this syntax to, e.g., Python.
  Everything discussed in this section is implemented in Coq, proven correct, and automatically extracted to an executable OCaml executable version.

  \paragraph*{Symbolic Evaluator and Symbolic Stores}
  
  Implementing a symbolic executor requires an efficient implementation of symbolic stores and a function to evaluate symbolic expressions.
  One way to implement symbolic stores is using hash tables with variable names as keys and expressions as values.
  Updating a store \(\sstore\) with an assignment \(x\coloneqq{}e\), written $\sstore[x \coloneqq e]$, returns an updated copy of the hash table.
  The symbolic evaluation \(\den{e}_\sstore\) of an expression $e$ in a store $S$ can be performed by a simple recursive function.
  We abstract from these implementation details and use the mathematical notations $S[x := e]$ and $\den{e}_S$ in code snippets to denote assignments and symbolic evaluations.

  \paragraph*{A Lazy Symbolic Executor}
  
  The main component of the symbolic executor is a function \code{expand} (\cref{fig:expand}) computing the successors of a symbolic state, implementing the symbolic semantics in \cref{fig:symbolic}.
  The executor's entry point is a function \texttt{run} (\cref{fig:run}) taking a FIFO queue of symbolic states and returning a stream of (optional) reachable symbolic states.
  At every iteration, it yields the first state from the queue.
  Then, the direct successors of the current state are computed using \code{expand} and enqueued in the task queue (with the lowest priority, realizing a breadth-first search).
  Finally, \texttt{run} calls itself recursively on the updated queue.
  If the queue is empty, a \texttt{None} token is yielded indicating that no states are left for expansion.

  \begin{figure}[tb]
    \centering
    \begin{minted}[escapeinside=//,mathescape=true,frame=single, framesep=5pt]{ocaml}
let run /$l$/ =
  match /$l$/ with
  | [] ->
    yield None
    run []
  | /$\sst{\varphi}{S}{p}$/::/$l$/ ->
    yield /$\sst{\varphi}{S}{p}$/
    run (/$l + \texttt{expand} \ \sst{\varphi}{S}{p}$/)
    \end{minted}
    \Description{Main Entry Point of the Symbolic Executor}%
    \caption{Main Entry Point of the Symbolic Executor}%
    \label{fig:run}
  \end{figure}

  \begin{figure}[tb]
    \centering
    \begin{minted}[escapeinside=//,mathescape=true, frame=single, framesep=5pt]{ocaml}
let expand /$\sst{\varphi}{S}{p}$/ =
  match /$p$/ with
  | /$\syn{skip}$/ | /$\syn{fail}$/ -> [ ]
  | /$\syn{skip}$/ /$\syn{;}$/ c ->
    [ /$\sst{\varphi}{S}{c}$/ ]
  | x /$\syn{=}$/ e ->
    [ /$\sst{\varphi}{S[\texttt{x} := \den{e}_S]}{\syn{skip}}$/ ]
  | /$c_1$/ /$\syn{;}$/ /$c_2$/ ->
    [ /$\sst{\varphi'}{S'}{c_1 \ \syn{;} \ c_2}$/ for /$\sst{\varphi'}{S'}{c_1}$/ in expand /$c_1$/ ]
  | /$\syn{while}$/ cond /$\syn{do}$/ c /$\syn{od}$/ ->
    [ /$\sst{\varphi \ \syn{and} \ \texttt{cond}}{S}{c \ \syn{;} \ p}$/, /$\sst{\varphi \ \syn{and} \ (\syn{and} \ \texttt{cond})}{S}{\syn{skip}}$/ ]
  | /$\syn{if}$/ /$cond$/ /$\syn{then}$/ /$c_1$/ /$\syn{else}$/ /$c_2$/ /$\syn{fi}$/ ->
    [ /$\sst{\varphi \ \syn{and} \ cond}{S}{c_1}$/, /$\sst{\varphi \ \syn{and} \ (\syn{not} \ cond)}{S}{c_2}$/ ]
    \end{minted}
    \Description{Computing Successors of Symbolic States}%
    \caption{Computing Successors of Symbolic States}%
    \label{fig:expand}
  \end{figure}

  \paragraph*{Reporting Bugs}
  
  The function \texttt{run} implements a symbolic executor.
  We obtain a bug finder by turning symbolic states from the stream into adequate \emph{status messages:}
  \begin{enumerate}
    \item $\texttt{BugFound}(\sst{\varphi}{V}{p})$ informs the user that a bug has been detected and returns the associated symbolic state.
    \item $\texttt{Pending}$ informs the user that the \gls{se} is in progress.
    \item $\texttt{Finished}$ informs the user that the \gls{se} terminated. No further symbolic state will be discovered.
  \end{enumerate}

  The conversion from symbolic states to status messages is accomplished by the function \texttt{display} (\cref{fig:display}).
  The final bug finder (\code{bug\_finder} in \cref{fig:display}) starts \gls{se} with a singleton task queue consisting of an initial symbolic state and applies \code{display} on the results.

  \begin{figure}[tb]
    \centering
    \begin{minted}[escapeinside=//,mathescape=true, frame=single, framesep=5pt]{ocaml}
let display state =
  if state is None then
    Finished
  else if is_stuck state then
    BugFound(state)
  else
    Pending

let bug_finder /$p$/ =
    map display (run [ /$\sst{\syn{true}}{id}{p}$/ ])
    \end{minted}
    \Description{Status Message Conversion, Bug Finder Entry Point}%
    \caption{Status Message Conversion, Bug Finder Entry Point}%
    \label{fig:display}
  \end{figure}

  \paragraph*{Detecting Bugs}
  
  The function \code{display} uses \code{is\_stuck} (\cref{fig:isstuck}) to filter error states. 
  As discussed in \cref{sub:sound_complete_bug_finding}, \code{is\_stuck} considers a symbolic state \(\sst{\pathcond}{\sstore}{\prog}\) an error state if \(p\) starts with $\syn{fail}$.

  \begin{figure}[tb]
    \centering
    \begin{minted}[escapeinside=//,mathescape=true, frame=single, framesep=5pt]{ocaml}
let is_stuck /$\sst{\varphi}{S}{p}$/ =
  match /$p$/ with
  | /$\syn{fail}$/ -> True
  | /$p \ \syn{;}$/ _ -> is_stuck /$p$/
  | _ -> False
    \end{minted}
    \Description{Classifying Error States}%
    \caption{Classifying Error States}%
    \label{fig:isstuck}
  \end{figure}

  \subsection{Relative Soundness and Completeness of the Symbolic Executor}%
  \label{sub:soundn-compln-se-engine}

  T-soundness and T-completeness are expressed assuming a ``yes/no'' oracle.
  For example, a T-sound bug finder only responds with \textit{yes} to \textit{Is there a bug in my program?} if there \emph{is} a bug that manifests for a concrete input.
  However, fuzzers and symbolic bug finders generally do not terminate unless interrupted, as in the case of our symbolic bug finder from the previous section.
  Thus, we adapt these notions to testers producing a \emph{stream} of results.
  We call the resulting concepts \emph{relative} soundness/completeness.
  They address the question of whether \emph{an input in a stream exposes a bug.}

  More precisely, a bug finder is \emph{relatively sound} if any bug report in the generated stream corresponds to bug-exposing input(s).
  It is \emph{relatively complete} if it will eventually report a bug for any given bug-triggering input.
  In other words, both terms are relative to the time provided to the bug finder.
  Moreover, relative completeness is parametric in a concrete bug-triggering input.
  Though formally weaker than T-completeness due to the additional precondition, it provides the strong guarantee a \emph{specific} bug will be found, not just \emph{any} bug.
  In the context of \gls{se}, relative \emph{soundness} additionally permits conveying false positives with \emph{unsatisfiable path constraints;}
  excluding these requires a complete constraint solver that correctly classifies all unsatisfiable path conditions.
  Provided such a solver, it is an easy exercise to construct a T-sound and T-complete bug finder from a relatively sound/complete one, as described in \cref{sub:tsound_tcomplete_bug_finder}.

  In this section, we investigate \gls{se}-specific notions of relative soundness and completeness, asserting that a \emph{symbolic executor} explores \emph{only} reachable symbolic states (soundness) and \emph{all} reachable states (completeness). 
  The subsequent \cref{sub:rel_sound_compl_bug_finder} formalizes the more general notions of relative soundness and completeness and applies the results from this section to prove that the bug finder from \cref{sec:implem} satisfies these properties.

%

  To conveniently express predicates over streams of symbolic states or status messages, we use \gls{ltl}~\cite[Chapter~5]{Baier2008} notation.
  \gls{ltl} formulas describe infinite streams of elements of some carrier set \(A\).
  Atomic \gls{ltl} formulas are predicates over \(A\)-elements, or, equivalently, subsets of \(A\).
  For example, the \gls{ltl} formula \(P=\{n\,\vert{}\,\exists m, 2\cdot{}m=n\}\) is satisfied by the stream \(4,5,6,\ldots\) since its \emph{first} element is even.
  The formula \(\always{}P\) (``globally'') describes streams in which \emph{all} elements are even; \(\eventually{}P\) (``eventually'') matches streams with \emph{some} even number.
  The property \(\always\eventually{}P\) holds for streams with infinitely many even numbers (``every position in the stream must eventually be succeeded by an even number'').
  In addition to ``\(\always\)'' and ``\(\eventually\),'' we use logical implication ``\(\ltlimplies\).''
  For example, \(\{42\}\ltlimplies(\eventually\{0,1\})\) means that any stream whose first element is 42 must contain either 0 or 1 at a later position.

  
  We write \(\stream\vDash\ltlformula\) if the stream \(\stream=\stream_0,\stream_1,\dots\) satisfies the \gls{ltl} formula \(\ltlformula\);
  $\stream\nvDash\ltlformula$ means that \(\stream\) does \emph{not} satisfy \(\ltlformula\).
  Formally, the semantics of LTL formulas is inductively defined as follows:
  \begin{alignat*}{2}
    \stream&\vDash  P & \ \Leftrightarrow \ & P(\stream_0)\\
    \stream&\vDash \syn{$\square$} \varphi & \ \Leftrightarrow \ &\forall i \ge 0, \stream_i\stream_{i + 1}... \vDash \varphi\\
    \stream&\vDash \syn{$\lozenge$} \varphi & \ \Leftrightarrow \ & \exists i \ge 0, \stream_i\stream_{i + 1}... \vDash \varphi\\
    \stream&\vDash \ltlformula_1\ltlimplies{}\ltlformula_2 & \ \Leftrightarrow \ & \stream \nvDash \ltlformula_1\text{ or }\stream\vDash\ltlformula_2
  \end{alignat*}

  Our symbolic executor produces a stream of optional (possibly \texttt{None}) symbolic states from a list \(l\) of input states.
  To be sound, every symbolic state in the stream must be \emph{reachable} from some state in \(l\).
  The \gls{ltl} formula $\reachableFrom(l)$ asserts that the first element in a stream is \texttt{None} or a symbolic state reachable from a state in $l$.

  \setCoqFilename{WiSE.implem.bugfinder_proof}
  \begin{definition}[Reachable From][reachable_from]
    \hspace*{10pt}$\reachableFrom(l) := \{ \sstate' \ | \ \exists \sstate \in l, \sstate \to^*_\textrm{sym} \sstate' \} \ \syn{$\vee$} \ \{ \texttt{None} \}$
  \end{definition}

  Our bug finder's soundness depends on the single-step execution function \texttt{expand} correctly implementing the relation \(\to_\textrm{sym}\):
  
  \begin{theorem}[Correct expansion][expand_sound]
    \hspace*{10pt}$\sstate \to_\textrm{sym}\sstate' \Leftrightarrow \sstate' \in \texttt{expand}(\sstate)$
  \end{theorem}
  \begin{proof}
    By induction on (1) the $\to_\textrm{sym}$-derivation (direction ``\(\Rightarrow\)'') and (2) the structure of $\sstate$'s program counter (``\(\Leftarrow\)'').
  \end{proof}

  Now we can state the relative soundness theorem, whose proof follows directly from \cref{coq:expand_sound}:

  \begin{theorem}[Relative Soundness of Symbolic Execution][run_sound]
    \hspace*{10pt}$\texttt{run}(l) \vDash \syn{$\square$}\reachableFrom(l)$
  \end{theorem}

  By outputting a \texttt{None} value in the stream, \texttt{run} signals that all reachable states have been \emph{exhaustively explored.}
  Consequently, the stream should only consist of \texttt{None} after the first issued one, guaranteeing that we can \textit{safely} terminate \gls{se} after the first \texttt{None}.

  \begin{theorem}[Sound Termination of Symbolic Execution][run_finished]
    \hspace*{10pt}$\texttt{run}(l) \vDash \syn{$\square$}(\texttt{None} \ \syn{$\to$} \ \syn{$\square$}\texttt{None}) $
  \end{theorem}

  We consider the symbolic executor \emph{relatively complete} if, for any initial symbolic state $\sstate$, it generates \emph{at least} all $\to^*_\textrm{sym}$-successors of $\sstate$ when started on the list $[\sstate]$.
  We first state a more general theorem (\cref{coq:run_steps_complete}) from which we then conclude the relative soundness.
  \cref{coq:run_steps_complete} asserts that for any state $\sstate$ discovered during \gls{se}, its (direct or indirect) successors will eventually be found.
  Its proof requires us to predict \emph{when exactly} the direct successors of any state \(\sstate\) will be discovered by our breadth-first search strategy.

  \begin{lemma}[Reachability of Successors][run_steps_complete]
    \hspace*{10pt}$\hat{\sigma} \to_\textrm{sym}^* \hat{\sigma}' \Rightarrow \forall l, \texttt{run}(l) \vDash \syn{$\square$}(\hat{\sigma} \ \syn{$\to$} \ \syn{$\lozenge$} \hat{\sigma}')$
  \end{lemma}
  \begin{proof}
    By induction on the length of the derivation $\sstate \to_\textrm{sym}^* \sstate'$.
    If $\sstate' = \sstate$, then the theorem is obvious. Otherwise, we have $s \to_\textrm{sym} \sstate''$ and $\sstate'' \to_\textrm{sym}^* \sstate'$ for some intermediate state $\sstate''$.
    Now suppose that $\sstate{}$ occurs at a given position $i_\sstate$ in the stream $\texttt{run(l)}$, 
    the direct successor of $\sstate''$ necessarily occurs at index $i_{\sstate''} = i_\sstate + |l| + 1$ or $i_{\sstate''} = i_\sstate + |l| + 2$ because \texttt{run} extends the task list by adding all the successors (see \cref{coq:expand_sound}) of the current state at the end of the task list and there are at most two successors for any state.
    Since $\sstate''$ occurs at position $i_{\sstate''}$, we know by induction hypothesis that $\sstate'$ occurs at some index $i_{\sstate''} + n$ for $n \ge 0$.
  \end{proof}

  From \cref{coq:run_steps_complete}, we conclude the completeness result.

  \begin{theorem}[Relative Completeness of Symbolic Execution][run_complete]
    \hspace*{10pt}$\hat{\sigma} \to_\textrm{sym}^* \hat{\sigma}' \Rightarrow \texttt{run}([\hat{\sigma}]) \vDash \syn{$\lozenge$} \hat{\sigma}'$
  \end{theorem}
  \begin{proof}
    It suffices to see that the first element of the stream $\texttt{run}([\sstate])$ is $\sstate$ and then to apply \cref{coq:run_steps_complete} at position 0.
  \end{proof}


  \subsection{Relative Soundness and Completeness of the Symbolic Bug Finder}
  \label{sub:rel_sound_compl_bug_finder}

  In the previous section, we established the relative soundness and completeness of the symbolic executor regarding the symbolic semantics.
  Now, we define \emph{relatively sound bug discovery}, a property relevant to bug finders in general, including fuzzers.
  We assume the interface of the function \findBugs{} from \cref{sec:implem}.
  That is, a bug finder outputs a stream of status messages that are either \texttt{Pending}, \texttt{Finished}, or \texttt{BugFound}.
  For a status message to be \emph{valid,} a \texttt{BugFound} message must report a symbolic state corresponding to a concrete set of bug-triggering states.
  The following definition of \validStatus{} uses a function \(\concrete(\sstate)\) mapping from a symbolic state to the set of concrete states it represents.
  Formally, \(\concrete(\sstate)=\{\cstate\,\vert\,\exists\cstore_0,\cstate\simeq_{\cstore_0}\sstate\}\).



  \begin{definition}[Valid Status][ValidStatus]
    \hspace*{10pt}$\validStatus{}(p) := \\\hspace*{20pt}\{ \texttt{BugFound}(\hat{\sigma}) \ | \ \forall \sigma, \sigma \in \textit{Concrete}(\hat{\sigma}) \Rightarrow \textit{IsBug}(p, \sigma) \}\\\hspace*{20pt} \syn{$\vee$} \ \texttt{Pending} \ \syn{$\vee$} \ \texttt{Finished}$
  \end{definition}
  
  According to this definition, a bug report is always valid for symbolic states with unsatisfiable path conditions since these have an empty \(\concrete(\sstate)\).
  How to get from there to real T-soundness is the focus of \cref{sub:tsound_tcomplete_bug_finder}.
  The definition of relative soundness is easy:
  A relatively sound bug finder only outputs valid status messages.
  
  \begin{theorem}[Relatively Sound Bug Discovery][relative_soundness]
    \hspace*{10pt}$\findBugs(p) \vDash \always\validStatus{}(p)$
  \end{theorem}
  \begin{proof}
    Direct application of \cref{coq:run_steps_complete}, 
    and the fact that \texttt{display} filters only stuck states.
  \end{proof}

  \paragraph*{Completeness}
  
  Relative completeness means that any given bug will eventually be reported in the stream of status messages.
  As mentioned before, this property is both \emph{weaker}---due to the precondition of providing a concrete bug-triggering input---and \emph{stronger}---since it guarantees that a pre-chosen bug will be found---than T-completeness.
  The subsequent completeness theorem uses a function \(\symbolic(\cstate)\) mapping from a concrete state to the set of symbolic states representing it (the inverse of \concrete).
  More formally, \(\symbolic(\cstate)=\{\sstate\,\vert\,\exists{}V_0,\cstate\simeq_{\cstore_0}\sstate\}\).

  \begin{theorem}[Relatively Complete Bug Discovery][relative_completeness]
    \hspace*{10pt}$\textrm{IsBug}(p, \sigma) \Rightarrow \texttt{find\_bugs}(p) \vDash \syn{$\lozenge$}\textrm{Symbolic}(\sigma)$
  \end{theorem}
  \begin{proof}
    Suppose $\textit{IsBug}(p, \sigma)$. By completeness of bug finding, we have $\textit{IsBug}_\textrm{sym}(p, \sigma)$. The result follows by the completeness of \texttt{run} and the fact that \texttt{display} keeps all stuck states in the stream.
  \end{proof}
  
  \paragraph*{Termination}

  Implementing a complete and yet always \emph{terminating} bug finder is impossible as it would solve the Halting problem.
  Yet, there are programs for which \gls{se} can be exhaustive (e.g., programs without loops nor recursion).
  We define ``sound termination'' as the property that whenever a bug finder emits a \texttt{Finished} message, it will not report anything else from that point on.
  If a soundly terminating, relatively complete bug finder does not report a bug and says \texttt{Finished}, we can be sure that \emph{the tested program is safe.}
  
  \begin{theorem}[Sound Termination of the Bug Finder][sound_termination]
    \hspace*{10pt}$\texttt{find\_bugs}(p) \vDash \syn{$\square$}(\texttt{Finished} \ \syn{$\to$} \ \syn{$\square$}\texttt{Finished})$
  \end{theorem}
  \begin{proof}
    Direct consequence of \cref{coq:run_finished}.
  \end{proof}

  \subsection{A T-Sound and T-Complete Bug Finder}
  \label{sub:tsound_tcomplete_bug_finder}

  A T-sound and T-complete bug finder is a \emph{binary oracle} answering \emph{yes} if, and only if, the tested program is faulty.
  So far, we considered bug finders producing infinite streams of status messages, to which T-soundness and T-completeness do not directly apply.
  Yet, relative soundness and completeness (\cref{coq:relative_completeness}, \cref{coq:relative_soundness}) are sufficiently strong to derive such an oracle.
  Relative soundness ensures for every $\texttt{BugFound}(\hat{\sigma})$ message that all $\cstate\in\textit{Concrete}(\hat{\sigma})$ are error states.
  The only obstacle in the way toward T-soundness is that $\textit{Concrete}(\hat{\sigma})$ could be empty, i.e., \(\sstate\)'s path condition \emph{unsatisfiable.}
  Otherwise, a \code{BugFound} message from the stream corresponds to a \emph{sound} bug.
  \gls{se} engines typically use off-the-shelf constraint solvers such as Z3 or CVC5 to determine if a path condition is unsatisfiable.

  \begin{figure}[tb]
    \centering
    \begin{minted}[escapeinside=//,mathescape=true, frame=single, framesep=5pt]{ocaml}
exception Bug
exception Termination

let report msg =
  match msg with
  | BugFound (path, _, _) ->
    if is_sat path then raise Bug
  | Finished -> raise Termination
  | _ -> ()

type answer = YES | NO

let has_bug p =
  try iter report (find_bugs p)
  with
    | Bug -> YES
    | Termination -> NO
    \end{minted}
    \Description{A T-sound and T-complete Bug Finder}%
    \caption{A T-sound and T-complete Bug Finder}%
    \label{fig:t_sound_oracle}
  \end{figure}
  
  Relative completeness, on the other hand, ensures that for any error state $\sigma$, the stream contains a message $\textit{BugFound}(\hat{\sigma})$ with $\hat{\sigma} \in \textit{Symbolic}(\sigma)$.
  The latter condition implies that \(\sstate\)'s path condition is satisfiable.
  Inconveniently, a \emph{terminating} bug finder can never be complete in general; the bug could always be hidden in the next state.
  Our solution to this problem is to implement a \emph{nonterminating but complete} symbolic tester.
  Our implementation, shown in \Cref{fig:t_sound_oracle}, builds on \code{find\_bugs} and uses a constraint solver \code{is\_sat} and a function \code{iter} for iterating over all elements of a stream.
  It only outputs ``sound'' bugs, terminates for finite symbolic executions, and continues \gls{se} as long as no bug is found or the process is interrupted.
  Since only terminating functions can be implemented in Coq, we have to resort to a ``pen-and-paper'' proof of $\texttt{has\_bug}$'s T-soundness and T-completeness.

  \begin{theorem}[T-soundness and T-completeness]
    \hspace*{10pt}$\texttt{has\_bug}(p) = YES \Leftrightarrow \textit{HasBug}(p)$
  \end{theorem}
  \begin{proof}
  The answer is $\texttt{YES}$ iff there is a $\texttt{BugFound}(\sst{\varphi}{S}{p'})$ message in the stream $\texttt{find\_bugs}(p)$ such that $\varphi$ is satisfiable. 
  The result follows from relative soundness/completeness.
  \end{proof}

  \section{Implementations and Case Studies}%
  \label{sec:case_studies}


  \rev{In this section, we connect to \cref{q2} from the introduction:
  \emph{How can we engineer a reliable symbolic testing tool?}}
  We aim to propose a design for symbolic executors independently of the implementation language;
  for the solution to be of general interest, it should be able to ``live'' outside an interactive theorem prover.
  Thus, we derived \emph{two executable implementations} from the Coq implementation of \wise.
  First, we automatically extracted OCaml code, wrapped in a command line interface with a parser for IMP's concrete syntax.
  Second, we developed a functionally equivalent Python implementation, \pywise{}, again with a usable command-line frontend.
  We used Python generators to implement streams where the Coq version uses co-induction \rev{and re-implemented the Coq code with only small, purely syntactical changes.}

  Since the transformations from Coq to executable programs are not formally proven correct, we evaluated them with three numeric algorithms.
  We chose the computation of a number's factorial and integer square root and the greatest common divisor of two numbers.
  We annotated each program with assertions of their correctness; e.g., for an integer square root \(r\) of a number \(n\), \(n\) must be in the closed interval \([r^2;(r+1)^2]\).
  Next, we derived a ``buggy'' mutation of each program, e.g., turning \code{\syn{while} \_s > x \syn{do}} into \code{\syn{while} \_s {\bfseries <} x \syn{do}} in the square root example.
  Finally, we asserted that our symbolic executors find the inserted bug but report no bug for the original programs.
  In addition, we verified that \gls{se} terminates with a \texttt{Finished} message for the correct programs when we restrict the domain of the input number(s) to a finite range.
  To demonstrate that a breadth-first path selection is superior to depth-first (as in the original DART symbolic executor~\cite{godefroid.klarlund.ea-05}) in finding \emph{any} bug in a program, we added a \code{-{}-depth-first} option to \pywise{} to transition to depth-first search.
  As a result, the bug in the factorial implementation is not uncovered since \gls{se} gets ``trapped'' in infinite iterations of an early loop.

  \rev{
  \paragraph*{Engineering Tasks}

  The main novelty of this paper is the presentation of a \emph{mechanized} specification of the correctness of testing tools and the implementation of a \emph{mechanically verified} symbolic testing tool in a proof assistant.
  We abstracted the addressed problem in two noteworthy ways to facilitate this project.

  \begin{enumerate}
    \item Our target programming language IMP is a \emph{simple WHILE language.}
          Extending our framework to a richer language requires the specification of a richer semantics.
          For example, the semantics of ES5 JavaScript in the JSCert project~\cite{bodin.chargueraud.ea-14} is implemented in about 3,000 lines of Coq code.
          The correctness proof of JSRef, an executable reference interpreter for JSCert, spans 3,500 lines of code. 
    \item We disregard the question of \emph{efficiency,} which is orthogonal to our~\cref{q1,q2}.
          Yet, efficiency is crucial for the competitiveness of \gls{se}-based testing tools compared to random testing~\cite{boehme.paul-16}.
          An efficient symbolic executor might replace the breadth-first path selection with, e.g., a coverage-guided generational search~\cite{godefroid.levin.ea-08}.
          A considerable body of literature on efficient symbolic execution exists (cf.\ the survey by Baldoni et al.~\cite{baldoni.coppa.ea-18} for an overview).
          The role of our contribution is to show how the correctness of an \gls{se} system can be retained in the face of efficiency optimizations.
          Exchanging path selection, for example, requires proving that all reachable symbolic states are \emph{eventually} considered by the new selection mechanism (\cref{coq:run_steps_complete}).
  \end{enumerate}

  Extending our foundational work to efficient systems for richer languages is possible.
  Of course, that step requires significant (proof) engineering work.
  This paper shows the way.
  }
  
  \section{Conclusion and Future Work}

  Considering that testing is the predominant program verification technique, it may seem surprising that the correctness of bug finders has not received more attention.
  Probably, this is because fuzzers, the most popular automated testing tools, are naturally ``bug-sound''---they only report inputs that \emph{have already} made the program under test crash.
  However, fuzzers only find shallow bugs for programs with complex demands on their inputs or with code guarded by complex constraints.
  \glsfirst{se} is a whitebox technique simulating multiple program executions at once.
  As such, this technique can handle complex input constraints fully automatically.
  Compared to a fuzzer, correctly implementing a symbolic executor constitutes a much more significant challenge.
  We defined when a tester is ``bug-sound'' and ``bug-complete'' (i.e., a program is buggy if, and only if, a bug is reported), resulting in the notions of T-soundness and T-completeness.
  We chose a semantics for \gls{se} and demonstrated that it can be used as a foundation for symbolic bug finding.
  Finally, we implemented a symbolic bug finder (\wise) in Coq, proved it T-sound and T-complete, and extracted executable implementations in OCaml and Python.

  The resulting symbolic executors interpret programs in the simple programming language IMP\@.
  Connecting to our foundational work, we plan to support additional IMP features, such as functions and pointers, in our symbolic executor.
  Thus, we could discover common pitfalls in implementing their symbolic evaluation.
  More generally, \wise{} constitutes a solid basis for studying more advanced \gls{se} features, such as different path selection algorithms, constraint representations, or constraint-solving approaches.

  Is it realistic to expect a fully verified symbolic executor for an industrial programming language?
  We think it is---considering the existence of fully verified \emph{C compilers}.
  In any case, blackbox testing will never find all the errors in existing symbolic executors;
  but insights from implementing different programming language features in a framework will positively influence their design.

  \section*{Data Availability}
 
  Our \wise{} and \pywise{} prototypes and the documentation of the Coq sources can be accessed at
  \ifreview
  \begin{center}
    \url{https://anonymous.4open.science/r/WiSE/}\\
    \url{https://anonymousverifier.github.io/wise_doc/toc.html}
  \end{center}
  \else
  \begin{center}
    \url{https://github.com/acorrenson/WiSE/}\\
    \url{https://acorrenson.github.io/WiSE/}\\
  \end{center}
  \fi
  \pywise{} is also available on PyPI and can be installed via
  \begin{tightcenter}
    \texttt{pip install wise-se}
  \end{tightcenter} 

\begin{acks}
This work was supported by the European Research Council (\grantsponsor{1}{ERC}{https://erc.europa.eu/}) Grant HYPER (No.~\grantnum{1}{101055412}) and by \grantsponsor{2}{DFG}{https://www.dfg.de/} grant \grantnum{2}{389792660} as part of TRR 248.
Views and opinions expressed are however those of the authors only and do not necessarily reflect those of the European Union or the European Research Council Executive Agency.
Neither the European Union nor the granting authority can be held responsible for them.
A. Correnson carried out this work as a member of the Saarbr\"ucken Graduate School of Computer Science.
\end{acks}

  \bibliographystyle{ACM-Reference-Format}
  \bibliography{references}
\end{document}